\documentclass[sigconf, nonacm]{acmart}
\usepackage{amsmath}
\usepackage{graphicx}
\usepackage{algorithm}
\usepackage{algpseudocode}
\usepackage{dsfont}
\usepackage{subcaption}
\usepackage{tikz}
\usetikzlibrary{shapes, arrows.meta, positioning, calc, backgrounds, fit, decorations.pathreplacing}

\usepackage{enumitem}

\newcommand\vldbdoi{XX.XX/XXX.XX}
\newcommand\vldbpages{XXX-XXX}
\newcommand\vldbvolume{14}
\newcommand\vldbissue{1}
\newcommand\vldbyear{2020}
\newcommand\vldbauthors{\authors}
\newcommand\vldbtitle{\shorttitle} 
\newcommand\vldbavailabilityurl{URL_TO_YOUR_ARTIFACTS}
\newcommand\vldbpagestyle{plain} 

\newcommand{\nikos}[1]{{\color{purple} [nikos: #1]}}

\newcommand{\stitle}[1]{\vspace{1ex}\noindent\textbf{#1}}

\begin{document}
\title{CRISP: Correlation-Resilient Indexing via Subspace Partitioning}

\author{Dimitris Dimitropoulos}
\affiliation{%
  \institution{U. of Ioannina \& Archimedes, Athena RC}
  \city{Ioannina}
  \state{}
  \country{Greece}}
\email{ddimitropoulos@cs.uoi.gr}

\author{Achilleas Michalopoulos}
\affiliation{%
  \institution{U. of Ioannina}
  \city{Ioannina}
  \state{}
  \country{Greece}}
\email{amichalopoulos@cs.uoi.gr}

\author{Dimitrios Tsitsigkos}
\affiliation{%
  \institution{Archimedes, Athena RC}
  \city{Athens}
  \state{}
  \country{Greece}}
\email{dtsitsigkos@athenarc.gr}

\author{Nikos Mamoulis}
\affiliation{%
  \institution{U. of Ioannina \& Archimedes, Athena RC}
  \city{Ioannina}
  \state{}
  \country{Greece}}
\email{nikos@cs.uoi.gr}


\begin{abstract}
  As the dimensionality of modern learned representations
increases to thousands of dimensions,
the state-of-the-art Approximate Nearest Neighbor (ANN) indices exhibit severe limitations.
Graph-based methods (e.g., HNSW) suffer from prohibitive memory consumption and routing degradation, while recent randomized quantization and learned rotation approaches (e.g., RaBitQ, OPQ) impose significant $O(N D^2)$ preprocessing overheads. 
We introduce CRISP,
a novel framework designed for ANN search in very-high-dimensional  spaces. Unlike rigid pipelines that apply expensive orthogonal rotations indiscriminately, CRISP employs a lightweight, correlation-aware adaptive strategy that redistributes variance only when necessary, effectively reducing the preprocessing complexity.
We couple this adaptive mechanism with a cache-coherent Compressed Sparse Row (CSR) index structure.
Furthermore, CRISP incorporates a multi-stage dual-mode query engine: a \textit{Guaranteed Mode} that preserves rigorous theoretical lower bounds on recall, and an \textit{Optimized Mode} that leverages rank-based weighted scoring and early termination to reduce query latency. Extensive 
evaluation on datasets of very high dimensionality (up to $D = 4096$) demonstrates that CRISP achieves state-of-the-art query throughput, low construction costs, and peak memory efficiency.

\end{abstract}

\maketitle

\pagestyle{\vldbpagestyle}
\begingroup\small\noindent\raggedright\textbf{PVLDB Reference Format:}\\
\vldbauthors. \vldbtitle. PVLDB, \vldbvolume(\vldbissue): \vldbpages, \vldbyear.\\
\href{https://doi.org/\vldbdoi}{doi:\vldbdoi}
\endgroup
\begingroup
\renewcommand\thefootnote{}\footnote{\noindent
This work is licensed under the Creative Commons BY-NC-ND 4.0 International License. Visit \url{https://creativecommons.org/licenses/by-nc-nd/4.0/} to view a copy of this license. For any use beyond those covered by this license, obtain permission by emailing \href{mailto:info@vldb.org}{info@vldb.org}. Copyright is held by the owner/author(s). Publication rights licensed to the VLDB Endowment. \\
\raggedright Proceedings of the VLDB Endowment, Vol. \vldbvolume, No. \vldbissue\ %
ISSN 2150-8097. \\
\href{https://doi.org/\vldbdoi}{doi:\vldbdoi} \\
}\addtocounter{footnote}{-1}\endgroup

\ifdefempty{\vldbavailabilityurl}{}{
\vspace{.3cm}
\begingroup\small\noindent\raggedright\textbf{PVLDB Artifact Availability:}\\
The source code, data, and/or other artifacts have been made available at \url{https://github.com/dabouledidia/CRISP}.
\endgroup
}

\section{Introduction}
Approximate Nearest Neighbor (ANN) search in a high-dimensional 
space is a core component of modern data systems, supporting applications ranging from content-based image retrieval to retrieval-augmented generation (RAG) for Large Language Models (LLMs). While recent
indexing approaches for ANN search are highly effective for medium-to-high dimensional data (e.g., $D \le 256$), scaling these indices to higher dimensions remains a significant challenge. Modern foundation models typically output vectors with very high dimensionalities, such as OpenAI's text embeddings ($D=3072$) or domain-specific descriptors like Trevi ($D=4096$). Scaling vector databases to manage such vectors  breaks the performance trade-offs of existing ANN algorithms \cite{10.1016/j.cogsys.2024.101216, DBLP:journals/pvldb/SunL0XRLN25}.

\looseness=-1
In particular,  graph-based methods, such as HNSW \cite{DBLP:journals/pami/MalkovY20}, which currently dominate industry deployments due to their high query throughput, 
in very-high-dimensional regimes ($D \ge 600$)
require substantial memory to store adjacency lists alongside the raw vectors, exhibit slow sequential index construction, and often have degraded routing efficiency on complex data distributions \cite{DBLP:journals/pvldb/WangXY021, Zhang2023VBASEUO}. 

\looseness=-1
To address memory limitations, recent research has pivoted toward methods based on subspace partitioning \cite{DBLP:journals/pacmmod/WeiLLPP25, Babenko2012, DBLP:journals/debu/0001C23} and quantization \cite{DBLP:journals/pacmmod/GaoL24, 10.1109/TPAMI.2010.57, DBLP:conf/cvpr/GeHK013, DBLP:conf/icde/PaparrizosELEF22}. Unlike graph methods, these approaches offer compact memory footprints through inverted index structures that scale linearly with dataset size and cache-efficient linear access patterns. Two prominent examples are SuCo (Subspace Collision) \cite{DBLP:journals/pacmmod/WeiLLPP25} and RaBitQ (Randomized Bit Quantization) \cite{DBLP:journals/pacmmod/GaoL24}. Nevertheless, applying these methods to thousands of dimensions reveals two practical limitations: they either struggle to process highly correlated features effectively, or their  preprocessing costs are too high.

SuCo's framework \cite{DBLP:journals/pacmmod/WeiLLPP25} assumes that the dimensions carry independent information, but real-world high-dimensional embeddings often violate this assumption. In correlated feature spaces (e.g., Gist and Fashion-MNIST), variance concentrates heavily in a small fraction of dimensions, causing subspaces to capture largely redundant information and rendering the collision proxy unable to discriminate between neighbors. This effect is illustrated in our evaluation (Figure~\ref{fig:recall_qps}), where retrieval quality of SuCo exhibits a clear recall ceiling on correlated 
datasets.

To handle highly correlated feature spaces, Optimized Product Quantization (OPQ) \cite{DBLP:conf/cvpr/GeHK013} and recent randomized quantization frameworks like RaBitQ \cite{DBLP:journals/pacmmod/GaoL24} apply global orthogonal rotations to redistribute variance across all dimensions. While effective, this transformation requires $O(ND^2)$ time for $N$ vectors of dimensionality $D$. This quadratic complexity to $D$ creates a significant preprocessing overhead for modern representation-learning models which produce embeddings with thousands of dimensions.
Since RaBitQ applies this transformation indiscriminately to all datasets, the $O(ND^2)$ overhead is paid 
even on naturally uncorrelated data where the original distribution is already suitable for indexing. 

\begin{figure}[htbp]
    \label{fig:crisp-arch}
    \centering
        \vspace{-0.1in}
\includegraphics[width=0.35\textwidth]{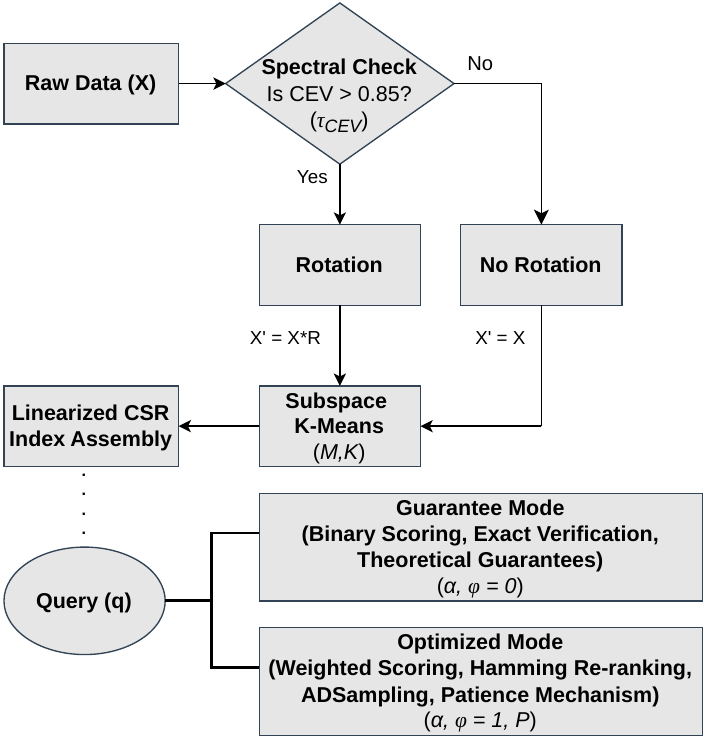}
    \vspace{-0.1in}    \caption{Overview of the CRISP Architecture.}
    \label{fig:CRISP-architecture}
\end{figure}

In this paper, we argue that a one-size-fits-all preprocessing pipeline is fundamentally incompatible with the demands of very-high-dimensional ANN search. We propose CRISP, an {\em adaptive} framework designed for feature spaces of very high dimensionality that dynamically adjusts its preprocessing strategy based on the correlation structure of the data, achieving robust retrieval, while whenever unnecessary bypassing the $O(ND^2)$ preprocessing cost.

CRISP's architecture is described in Figure \ref{fig:CRISP-architecture}.
In its preprocessing phase, CRISP employs a lightweight heuristic that analyzes data correlation to decide the necessity of variance redistribution.
CRISP fundamentally bridges two complementary paradigms: it inherits the low indexing cost of collision-based methods \cite{DBLP:journals/pacmmod/WeiLLPP25} and the distributional robustness of quantization approaches \cite{DBLP:journals/pacmmod/GaoL24}. 

After applying the rotation decision mechanism and potentially applying rotation, CRISP partitions the feature space into $M$ disjoint subspaces. Within each subspace, we employ an Inverted Multi-Index (IMI) partitioning scheme by decomposing the subspace into lower-dimensional segments and performing independent k-means clustering on each to create a fine-grained grid of cells. This approach builds on the subspace collision paradigm introduced by SuCo as a proxy for Euclidean distance, while introducing optimizations specifically tailored for the high-dimensional and correlated distributions where existing indices fail to scale.

We couple the subspace partitioning structure with a highly optimized system architecture designed to maximize cache efficiency and instruction-level parallelism. We introduce a cache-coherent Compressed Sparse Row (CSR) index structure for each subspace, replacing the pointer-chasing overhead of traditional hash-based inverted lists with sequential memory access patterns. In our implementation, all point IDs for a given subspace are stored in a single contiguous array, while a secondary offsets array explicitly marks the start and end of each partition bucket. This design significantly reduces Translation Lookaside Buffer (TLB) misses during the candidate voting and collision counting phase, where the system must rapidly retrieve and aggregate point IDs from activated buckets. By linearizing the index, hardware prefetchers can stream bucket contents efficiently into cache, resulting in substantial query speedups in very-high-dimensional regimes.

To maximize query efficiency, we design a multi-stage query pipeline that progressively filters the dataset into a candidate set. The process begins 
by applying 
subspace collision counting to retrieve the set of points most likely to contain the true nearest neighbors. We then apply dual-mode verification to these filtered candidates. In Guaranteed Mode, the system performs exhaustive exact Euclidean verification on the candidate set, preserving rigorous theoretical lower bounds on recall derived from subspace collision statistics. In Optimized Mode, we prioritize throughput by introducing a rank-based weighted scoring mechanism. This refines the coarse collision counting used in SuCo into a distance-aware ranking by assigning higher importance to collisions occurring in proximal subspaces. This mode further accelerates queries through Hamming distance re-ranking, ADSampling, and a dynamic patience termination mechanism.

Our contributions can be summarized as follows:
\begin{itemize}[leftmargin=15pt]
\item \textbf{Adaptive preprocessing strategy:} We introduce a correlation-aware preprocessing mechanism,
   which selectively applies randomized orthogonal rotations only when high inter-dimensional correlation is detected. On uncorrelated data, preprocessing is bypassed entirely, avoiding the $O(ND^2)$ preprocessing overhead. This adaptive strategy contrasts with RaBitQ and OPQ's indiscriminate application of expensive global transformations.
    \item \textbf{Rigorous theoretical guarantee:} We derive a conditional recall lower bound using Hoeffding’s \cite{hoeffding1963} inequality, proving that retrieval failure probability decays exponentially with the number of subspaces. This is a strictly tighter guarantee than the polynomial bounds (Chebyshev) provided by prior subspace collision frameworks \cite{DBLP:journals/pacmmod/WeiLLPP25}, ensuring robust retrieval quality when our adaptive independence condition is met.
    \item \textbf{Multi-stage dual-mode query engine:} We design a query pipeline for efficient candidate ordering, coupled with mode-specific verification strategies. Guaranteed Mode performs exhaustive  verification to preserve rigorous recall lower bounds. Optimized Mode instead employs three complementary acceleration techniques: (i) rank-based weighted scoring, (ii) ADSampling \cite{DBLP:journals/pacmmod/GaoL23}, and (iii) dynamic patience termination to maximize throughput.
    \item \textbf{Evaluation on Very-High-Dimensional Vectors:} We rigorously evaluate CRISP against the industry-standard graph baseline (HNSW), learned rotation (OPQ), randomized quantization (RaBitQ), and the original subspace collision framework (SuCo) on datasets of dimensionality up to $D=4096$. We show that CRISP 
    achieves superior Pareto-optimal trade-offs for throughput, recall, and construction time, while maintaining a  linear memory footprint where graph methods fail to scale.
\end{itemize}

\section{Related Work}

The problem of Approximate Nearest Neighbor (ANN) search in high-dimensional spaces has been extensively studied, with approaches generally categorized into three primary families: parti\-tioning-based, quantization-based, and graph-based methods.

\subsection{Partitioning and Subspace Methods} 
\label{sec:partition}
Early approaches to ANN search focused on space partitioning to prune candidates efficiently. The Vector Approximation File (VA-File) \cite{DBLP:conf/vldb/WeberSB98} pioneered the use of scalar quantization to compress vectors and filter candidates based on lower bounds. The VA-file evolved into sophisticated partition-based indices, such as the Inverted Multi-Index (IMI) \cite{Babenko2012}, which decomposes the vector space into orthogonal subspaces and applies independent k-means clustering within each subspace to generate a fine-grained Cartesian product of cells. Each data point is assigned to its nearest centroid in every subspace, and at query time, the system identifies the closest centroids to the query in each subspace and retrieves the union of their posting lists as candidate neighbors.

Building on this partitioning paradigm, where subspace structure is used primarily for candidate pruning rather than compact vector encoding, a series of works have sought to refine how subspace proximity is measured and exploited. Query-Aware LSH (QALSH) \cite{huang2015query} and PM-LSH \cite{DBLP:journals/pvldb/ZhengZWHLJ20} improve collision bounds through subspace projections, though both encounter load balancing challenges under highly skewed distributions. SuCo \cite{DBLP:journals/pacmmod/WeiLLPP25} pushes this direction further by introducing a collision-counting metric as a proxy for Euclidean distance. However, all these approaches inherit a common limitation: they assume that local subspace proximity implies global proximity, an assumption violated by the strong inter-dimensional correlations inherent to real-world very-high-dimensional data, as we demonstrate in Section \ref{sec:exp}. In such regimes, subspaces capture redundant information, causing the collision proxy to fail.

\subsection{Quantization-based Methods} 
Unlike  partitioning methods,
which use subspace decomposition primarily for candidate pruning, quantization-based methods exploit the same decomposition principle to compress vectors into compact codes, reducing memory footprint at scale. Product Quantization (PQ) \cite{10.1109/TPAMI.2010.57} pioneered the decomposition of vectors into sub-spaces for local codebook learning, a technique foundational to established high-performance libraries such as FAISS \cite{johnson2019billion}. To reduce quantization error on correlated data, Optimized Product Quantization (OPQ) \cite{DBLP:conf/cvpr/GeHK013} learns a parametric rotation matrix to minimize distortion. However, OPQ relies on iterative optimization, which becomes too expensive as dimensionality increases. For ultra-high-dimensional vectors ($D \ge 1000$), optimizing a dense $D \times D$ rotation matrix incurs a significant training bottleneck.

\looseness=-1
Most recently, RaBitQ \cite{DBLP:journals/pacmmod/GaoL24} proposed a randomized quantization framework utilizing a global orthogonal rotation to flatten the vector energy spectrum, thereby providing theoretical error bounds for bitwise distance estimation. However, the $O(N \cdot D^2)$ complexity of the rotation
constitutes a substantial preprocessing cost for high values of $D$. Hybrid approaches, such as ScaNN \cite{DBLP:conf/icml/Guo0L0CSK20} and RoarGraph \cite{ DBLP:journals/pvldb/ChenZHJW24}, combine quantization with anisotropic scoring or graph structures. ScaNN addresses quantization error through specialized loss functions and complex query-side scoring, whereas RoarGraph inherits the high construction costs associated with graph indexing.

\subsection{Graph-based Methods} 
Graph-based indices, particularly the Hierarchical Navigable Small World (HNSW) \cite{DBLP:journals/pami/MalkovY20}, are the current industry standard for in-memory ANN search \cite{DBLP:journals/debu/0001C23}. These methods rely on a proximity graph where the search algorithm navigates greedily toward the target. Extensions like Vamana and DiskANN \cite{DBLP:conf/nips/SubramanyaDK0K19} adapt this structure for SSD storage to handle larger data. However, benchmarks show that graph performance degrades significantly as dimensionality increases \cite{DBLP:journals/pvldb/WangXY021}, since greedy routing becomes less effective as the distance differences between neighbors diminishes \cite{DBLP:conf/icdt/BeyerGLS99, DBLP:journals/pvldb/WangWCWPW24}. This fundamental limitation is exacerbated by the fact that traditional distance comparisons become statistically unreliable and computationally expensive at scale, a challenge that necessitates more robust distance comparison operations \cite{DBLP:journals/pacmmod/GaoL23}. Furthermore, storing the graph structure alongside high-dimensional vectors creates an unmanageable memory footprint \cite{10.14778/3611479.3611537}. While systems like SPANN \cite{DBLP:conf/nips/ChenZWLLLYW21} reduce memory usage by offloading data to disk, they incur high latency from disk reads during their final search stage \cite{DBLP:conf/nips/ChenZWLLLYW21, DBLP:journals/pacmmod/WangXYWPKGXGX24}.

\section{Preliminaries}
\label{sec:preliminaries}
In this section, we formalize the problem of ANN search and define the core primitives that form the basis of our adaptive framework. Table~\ref{tab:notations} summarizes the key notation used throughout the paper.

\begin{table}[t]
  \caption{Summary of Key Notations}
  \label{tab:notations}
\vspace{-2mm}
  \centering
  \small
  \begin{tabular}{c l}
    \toprule
    \textbf{Symbol} & \textbf{Description} \\
    \midrule
    $N$ & Dataset cardinality \\
    $D$ & Data dimensionality \\
    $M$ & Number of subspaces \\
    $K$ & Number of centroids per subspace \\
    $k_{\text{size}}$ & Number of true nearest neighbors \\
    $X$ & Original dataset matrix ($N \times D$) \\
    $C$ & Candidate set filtered by collision threshold $\tau$  \\
    $R$ & Randomized orthogonal rotation matrix ($D \times D$) \\
    $\tau_{CEV}$ & Cumulative Explained Variance threshold \\
    $\phi$ & Dual-mode execution flag \\
    $\alpha$ & Minimum required subspace collision percentage \\
    \bottomrule
  \end{tabular}
\end{table}

\stitle{Exact and Approximate Similarity Search.} Similarity search, often referred to as Nearest Neighbor (NN) search,  finds the most similar items in a database to a specified query object, based on a distance function $\operatorname{dist}(\cdot)$. Formally, let $X \in \mathbb{R}^{N \times D}$ be a dataset matrix comprising $N$ real-valued $D$-dimensional vectors, and let $q \in \mathbb{R}^D$ be a query vector. The objective is to identify the data vector in $X$ that minimizes the distance to the query:
\[
    \text{NN}(q) = \operatorname*{arg\,min}_{x \in X} \operatorname{dist}(q, x)
\]
Euclidean distance (L2) is the standard distance function choice in most applications~\cite{li2020approximate,wang2018survey}. This problem formulation naturally extends to $k$-Nearest Neighbor ($k$NN) search, which seeks to retrieve the set of $k$ items closest to the query.

\looseness=-1
Similarity search algorithms are generally classified into two categories: exact and approximate~\cite{wang2016learning}. Exact methods guarantee the retrieval of the true nearest neighbors; however, they often incur high computational costs that render them unsuitable for applications that require high query throughput. Conversely, ANN methods trade a marginal loss in retrieval accuracy for significant gains in query throughput. In most large-scale, high-dimensional scenarios, approximate solutions are preferred as they enable low-latency processing while maintaining sufficient accuracy~\cite{li2020approximate, wang2016learning}. Balancing this trade-off between retrieval precision and efficiency is the central challenge in the design of modern ANN systems~\cite{aumuller2017ann}.

\stitle{Subspace Collision Framework.} The Subspace Collision (SuCo) framework \cite{DBLP:journals/pacmmod/WeiLLPP25} approximates the proximity between two vectors by decomposing the $D$-dimensional space into $M$ disjoint subspaces. In each subspace $m$, an IMI scheme further splits the subspace into two halves, learning a codebook of $K$ centroids per half via k-means. Each data point $x$ is represented by its nearest centroid in both 
halves of every subspace, mapping it to a discrete tuple of cell indices $u(x) = [u_1(x), \dots, u_M(x)]$. During retrieval, grid cells are explored in ascending order of their aggregated Euclidean distance to the query sub-vector, using a Dynamic Activation algorithm (a variant of the Multi-Sequence algorithm) until a sufficient candidate set is retrieved. The proximity between a query $q$ and a point $x$ is estimated by the {\em collision count} $S_{\text{col}}(q, x)$: the number of subspaces in which $x$'s assigned cell is activated during retrieval. Formally:

\begin{equation}
    S_{\text{col}}(q, x) = \sum_{m=1}^{M} \mathds{1} \left( u_m(x) \in \mathcal{U}_m(q) \right)
\end{equation}

This integer score serves as a 
probabilistic estimator
for Euclidean distance, effectively quantizing proximity into $M+1$
discrete levels. 
Existing frameworks provide theoretical guarantees on retrieval failure via Chebyshev's inequality, though these polynomial bounds assume that collisions across subspaces are independent events. This assumption is violated when features across different subspaces are correlated, reducing the SuCo's discriminative power.

\stitle{Candidate Refinement.} Refinement refers to the re-ranking of the candidate set $C$ identified during the subspace collision filtering stage. 
Since the initial collision counts are discrete integers (and thus provide low-resolution ranking), refinement employs higher-precision distance estimators, such as Hamming distance on binary signatures or ADSampling on float vectors, to strictly order and prune candidates before the final top-$k$ selection.

\stitle{Randomized Orthogonal Rotation.} To mitigate the quantization error caused by uneven variance distribution in correlated data, RaBitQ employs a randomized orthogonal transformation. Let $G \in \mathbb{R}^{D \times D}$ be a random Gaussian matrix with entries $G_{ij} \sim \mathcal{N}(0,1)$. The rotation matrix $R$ is obtained via QR decomposition of $G$: $G = Q \cdot R'$, where we set $R \leftarrow Q$. The transformed dataset $X' = \{x \cdot R \mid x \in X\}$ exhibits the property that the variance of each dimension is approximately equalized. 
While this ensures that subsequent quantization steps operate on statistically well-behaved data, applying this global transformation requires a dense matrix multiplication with complexity $O(N D^2)$. For very-high-dimensional datasets (e.g., $D=4096$), this preprocessing step imposes a prohibitive computational overhead. In CRISP, we utilize this transformation adaptively, applying it only when the inter-dimension correlation (measured against $\tau_{CEV}$) necessitates it.

\stitle{Binary Quantization.} Binary Quantization (BQ) \cite{Gong2011} maps continuous vectors to compact bitstrings $b(x) \in \{0,1\}^D$ via a sign function, producing a highly compressed vector representation. The dissimilarity between two binary codes is efficiently measured using the Hamming distance $\|\cdot\|_H$, which can be computed using fast hardware-level bitwise XOR and popcount operations. In CRISP, BQ serves as a lightweight proxy for Euclidean proximity, used to sort and prioritize the candidate set $\mathcal{C}$ during the refinement stage before expensive distance computations are performed.

\stitle{ADSampling.} ADSampling \cite{DBLP:journals/pacmmod/GaoL23} estimates the squared Euclidean distance between a query $q$ and a candidate $x$ using an incrementally increasing subset of $t$ dimensions. To prevent the premature pruning of true nearest neighbors, the framework adopts a pruning condition based on a relative error bound:
\begin{equation}
d_t^2(q, x) > r_k^2 \cdot \frac{t}{D} \left(1 + \frac{\epsilon_0}{\sqrt{t}}\right)^2
\end{equation}
where $r_k^2$ is the squared distance to the current $k$-th nearest neighbor and $d_t^2(q, x) = \sum_{i=1}^{t} (q_i - x_i)^2$ is the partial sum. CRISP sets the safety margin $\epsilon_0 = 2.1$ following the implementation in \cite{DBLP:journals/pacmmod/KuffoKB25}. The engine performs these checks at fixed intervals of $t = 32$ dimensions. This allows the refinement stage to progressively discard non-neighbor candidates with high statistical confidence, significantly reducing the number of full-dimensional distance computations required.

\section{The CRISP framework}
\looseness=-1
This section presents CRISP, our adaptive indexing framework which reconciles the efficiency of subspace partitioning with the robustness of randomized quantization. CRISP, designed for 
very high dimensional datasets 
($D \geq 600$),  consists of three distinct phases:
\begin{enumerate}[leftmargin=15pt]
    \item \textbf{Correlation-Aware Preprocessing}, which adaptively transforms the data distribution to restore subspace independence.
    \item \textbf{CSR-Based Indexing}, which maps the data into a cache-coherent inverted structure to maximize memory throughput.
    \item \textbf{Dual Multi-Stage Query Pipeline}, which progressively filters candidates using vectorized operations.
\end{enumerate}

\noindent An overview of CRISP architecture is provided in Figure \ref{fig:CRISP-architecture}. Its query processing steps are summarized by Algorithm \ref{alg:dual_query}.

\subsection{Correlation-Aware Preprocessing}
\label{sec:cev}
We observe that subspace collision methods fail primarily on correlated datasets, where the energy of the vectors is concentrated in a few principal components. In high-dimensional spaces, this concentration causes subspaces to capture redundant information, which reduces the discriminative power of the index. To address this, we introduce a lightweight \textbf{Spectral Correlation Check} prior to indexing. To ensure that this step incurs negligible overhead, we compute the covariance matrix on a bounded random sample $X_{sample} \subset X$ of empirically set size to $\min(0.1N, 10^5)$. We perform Eigenvalue Decomposition to calculate the \textbf{Cumulative Explained Variance  (CEV)}, defined as the cumulative variance explained by the top 20\% of principal components:

\[
\text{CEV} = \frac{\sum_{i=1}^{k} \lambda_i}{\sum_{j=1}^{D} \lambda_j}, \quad \text{where } k = \lfloor 0.2 \cdot D \rfloor
\]

Based on this spectral analysis, the framework dynamically selects the optimal construction path. We utilize an experimentally derived threshold of $\tau_{CEV}=0.85$ (justified via an ablation study in Subsection \ref{sec:ablation}).
This value serves as a conservative proxy for identifying data distributions where inter-dimension correlation is strong enough to invalidate our theoretical error bounds, which we define and analyze formally in Section ~\ref{sec:theorem}. When the CEV exceeds this threshold, CRISP detects a significant deviation from isotropy and triggers "Variance Redistribution". In this case we apply the randomized orthogonal rotation $R$ defined in Section ~\ref{sec:preliminaries}, generating a transformed dataset $X'=XR$. This transformation effectively redistributes variance uniformly across the vector space, ensuring no single dimension dominates the partitioning process.

Conversely, for distributions where the CEV falls below $0.85$,  CRISP concludes that the natural distribution is sufficiently dispersed. In such cases, we bypass the rotation step entirely, thereby avoiding the $O(N D^2)$ computational overhead and proceeding directly to indexing.
This integrated design offers a distinct architectural advantage over rigid quantization frameworks like RaBitQ. While RaBitQ treats rotation as a decoupled, external preprocessing step, requiring the materialization of a transformed dataset copy, CRISP persists the rotation matrix directly within the index metadata. This enables the query engine to toggle between native and rotated modes without external dependencies. 

Crucially, this integration ensures superior memory scalability. While decoupled pipelines (e.g., RaBitQ) typically treat rotation as an external preprocessing step that requires the materialization of a second $N \times D$ dataset, CRISP performs the transformation in-place during index construction. This is achieved by iterating through the dataset vector-by-vector and utilizing a small thread-local buffer of size $D$ to store intermediate dot products before overwriting the original memory addresses. Consequently, peak memory usage in CRISP never exceeds the raw dataset size ($ND$), whereas separated pipelines often reach double this size ($2ND$).

\begin{algorithm}[t]
\caption{CRISP's Dual-Mode Query Execution}
\label{alg:dual_query}
\begin{algorithmic}[1]
\footnotesize
\Require Query $q$, Mode $\phi$, Index $\mathcal{I}$
\Ensure Top-$k$ Candidates
\State Initialize score map $\mathcal{V} \gets \mathbf{0}$
\For{$m \gets 1$ to $M$}
    \State Decompose $q^{(m)}$ into $q_{left}, q_{right}$ 
    \State $Dist_1 \gets \text{Dist}(q_{left}, C_{left})$; $Dist_2 \gets \text{Dist}(q_{right}, C_{right})$ 
    \State $\mathcal{Q} \gets \text{PriorityQueue}(\{(Dist_1[0]+Dist_2[0], 0, 0)\})$ 
    \State $\text{rank} \gets 0$; $\text{retrieved} \gets 0$ 
    \While{$\text{retrieved} < \text{budget}$} 
        \State $(cost, i, j) \gets \mathcal{Q}.\text{pop}()$ 
        \State $cell \gets \text{Combine}(i, j)$; $\text{rank} \gets \text{rank} + 1$ 
        \State $w \gets 1$ 
        \If{$\phi = 1$ \textbf{and} $\text{rank} \le k_{size}$} 
            \State $w \gets 2$ 
        \EndIf
        \For{$id \in \mathcal{I}.\texttt{ids}_m[cell]$} 
            \State $\mathcal{V}[id] \gets \mathcal{V}[id] + w$ 
            \State $\text{retrieved} \gets \text{retrieved} + 1$ 
        \EndFor
        \State $\mathcal{Q}.\text{push}(\text{next candidates from } Dist_1, Dist_2)$ 
    \EndWhile
\EndFor
\State Candidates $\mathcal{C} \gets \{x \mid \mathcal{V}[x] \ge \tau\}$ 
\If{$\phi = 1$} 
    \State $\mathcal{C} \gets \text{SortByHamming}(\mathcal{C})$ 
\EndIf
\For{$x \in \mathcal{C}$} 
    \State \textbf{Verify} using $\text{ADSampling}(\phi=1)$ or $\text{ExactL2}(\phi=0)$ 
    \If{$\phi=1$ \textbf{and} \text{PatienceReached}} \State \textbf{break}  
    \EndIf
\EndFor
\State \Return Top-$k(\mathcal{C})$ 
\end{algorithmic}
\end{algorithm}

\subsection{Cache-Coherent CSR Indexing}
\label{sec:csr}

\looseness=-1
Retrieval of high dimensional data is fundamentally memory bandwidth bound. As dimensionality increases, larger vector sizes cause traditional Inverted File (IVF) structures to trigger frequent CPU stalls and cache misses. IVF implementations in previous work store posting lists in fragmented data structures that may incur ``pointer-chasing'' overhead and 
TLB misses.

For example, SuCo \cite{DBLP:journals/pacmmod/WeiLLPP25} implements its index using a vector of unordered maps
({\footnotesize\texttt{vector<unordered\_map<pair<int,~int>,~vector<int>>>}}), \linebreak
where retrieving the IDs for a subspace cell requires a hash-table lookup
followed by an access to an independently heap-allocated vector, resulting
in scattered, non-contiguous memory accesses. RaBitQ \cite{DBLP:journals/pacmmod/GaoL24} similarly suffers from
poor spatial locality, as its search function must aggregate distance
components from multiple independent memory regions per cluster probe.
To address these ``pointer-chasing'' overheads, CRISP employs a Compressed
Sparse Row (CSR) structure that linearizes the index into a single contiguous
memory block, as illustrated in Figure \ref{fig:layouts}. 

CRISP's CSR implements the inverted index as a sparse matrix where rows correspond to cells and columns to data points. A \textit{cell} is defined as the Cartesian product of two sub-centroid sets ($K \times K$) within a subspace $m$. CRISP sets the number of centroids per sub-partition to $K=50$ based on \cite{DBLP:journals/pacmmod/WeiLLPP25}. During construction, each data point $x_n \in \{1, \ldots, N\}$ is assigned to a specific cell. To ensure physical locality, we generate $(cell, x_n)$ tuples and strictly sort the dataset based on the cell identifier. For example, if points $\{2, 4, 7, 9, 15, 22\}$ are assigned to cell $C_{i}$ and points $\{3, 8, 12, 20\}$ to cell $C_{j}$, they are reordered in memory as $\{3, 8, 12, 20, 2, 4, 7, 9, 15, 22\}$. This ensures that all identifiers belonging to the same cell are stored contiguously, eliminating the pointer-chasing overhead found in fragmented implementations. The resulting index is consolidated into two per-subspace arrays:

\begin{itemize}
[leftmargin=15pt]    \item \texttt{Offsets}: An array of size $K^2+1$, where $K$ denotes the number of centroids per sub-partition. The values \texttt{Offsets[i]} and \texttt{Offsets[i+1]} serve as integer offsets that delimit the start and end of the $i$-th cell's posting list within the data array.

    \item \texttt{Vectors IDs}: A single contiguous array of size $N$ per subspace with the sorted point identifiers in a cache-aligned block.
\end{itemize}

As shown in Algorithm ~\ref{alg:dual_query}, 
 this structure allows the query execution to iterate through a specific cell's candidates by streaming a single contiguous memory segment. By resolving segment boundaries in constant time via the \texttt{Offsets} array, CRISP maximizes hardware prefetching efficiency and shifts the bottleneck from memory latency to memory bandwidth.
CRISP's contiguous layout of candidate identifiers within the \texttt{ids} array 
allows the accumulation phase to approach theoretical memory throughput.

\begin{figure}[t]
    \centering
    \begin{subfigure}{\columnwidth}
        \centering
        \includegraphics[width=0.99\columnwidth]{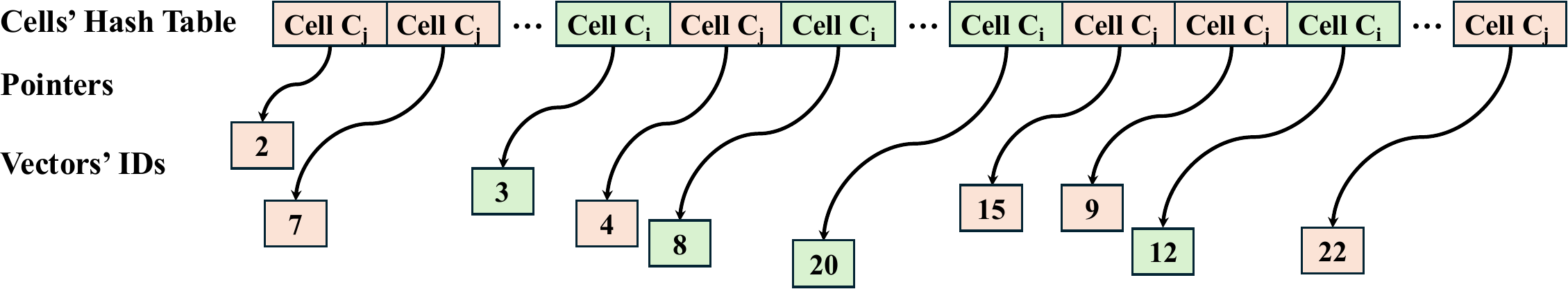}
        \caption{Traditional Fragmented Layout.}
        \label{fig:hash_based}
    \end{subfigure}
    
    \vspace{0.4cm} 
    
    \begin{subfigure}{\columnwidth}
        \centering
        \includegraphics[width=0.99\columnwidth]{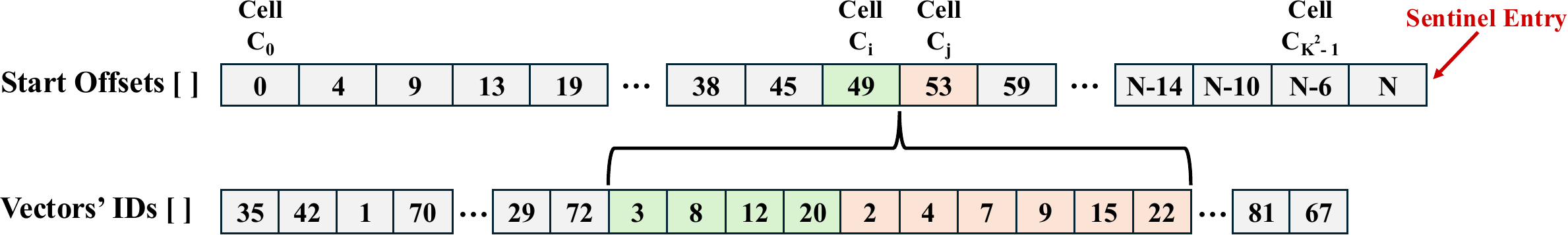}
        \caption{CRISP CSR Layout.}
        \label{fig:crisp_csr}
    \end{subfigure}
    \vspace{-0.2in}
    \caption{CRISP replaces fragmented hash-based traversal (a) by a Compressed Sparse Row (CSR) layout (b), bottleneck shift from memory latency to peak memory bandwidth.}
    \label{fig:layouts}
\end{figure}

\subsection{Multi-Stage Dual-Mode Query Engine}
\vspace{1em}
\subsubsection{Adaptive Collision Scoring}
To bridge the gap between theoretical rigor and performance in practice, CRISP uses a dual-strategy scoring mechanism determined by the execution mode $\phi$, as implemented in the query execution flow of Algorithm ~\ref{alg:dual_query}.

\stitle{Binary Scoring (Guaranteed Mode, $\phi$=0).} When {\em strict theoretical guarantees are required}, we utilize Binary Collision Counting. A data point $x$ receives a score increment of $+1$ if and only if it falls within the assigned centroid cell of query $q$ in subspace $S_i$. This ensures the total collision score follows a Binomial distribution, satisfying the independence assumptions required for the lower-bound proof of Theorem  \ref{theorem:conditionalrecall} in Section ~\ref{sec:theorem}.

\stitle{Weighted Scoring (Optimized Mode, $\phi$=1).} In the latency-critical optimized path, we relax the binary constraint to capture finer-grained proximity information using a Rank-Based Weighted Scheme. 
In contrast to former subspace collision and collision-counting frameworks \cite{DBLP:journals/tkde/TianZZ23, DBLP:journals/pvldb/ZhengZWHLJ20, DBLP:journals/pvldb/WeiPLP24, DBLP:journals/pacmmod/WeiLLPP25}, which typically employ the uniform binary scoring system (as in CRISP's Guaranteed Mode), CRISP's Optimized Mode introduces a proximity-aware weighting scheme. By prioritizing candidates found in these highest-ranked regions, CRISP allows high-quality neighbors to reach the collision threshold significantly faster than other methods.

Specifically, during the query phase, for each subspace $m$, we explore the nearest cells to the query sub-vector $q(m)$ in a specific order. We assign a collision weight $W$ based on the rank of the cell being visited. Formally, let $rank(u_m(x))$ denote the visit order of the cell assigned to point $x$. The weighting function is defined as:

\[
\mathcal{W}(q^{(m)}, x^{(m)}) = 
\begin{cases} 
2 & \text{if } rank(u_m(x)) \leq k_{size} \\
1 & \text{otherwise}
\end{cases}
\]

\begin{figure}[t]
    \centering
    \includegraphics[width=0.47\textwidth]{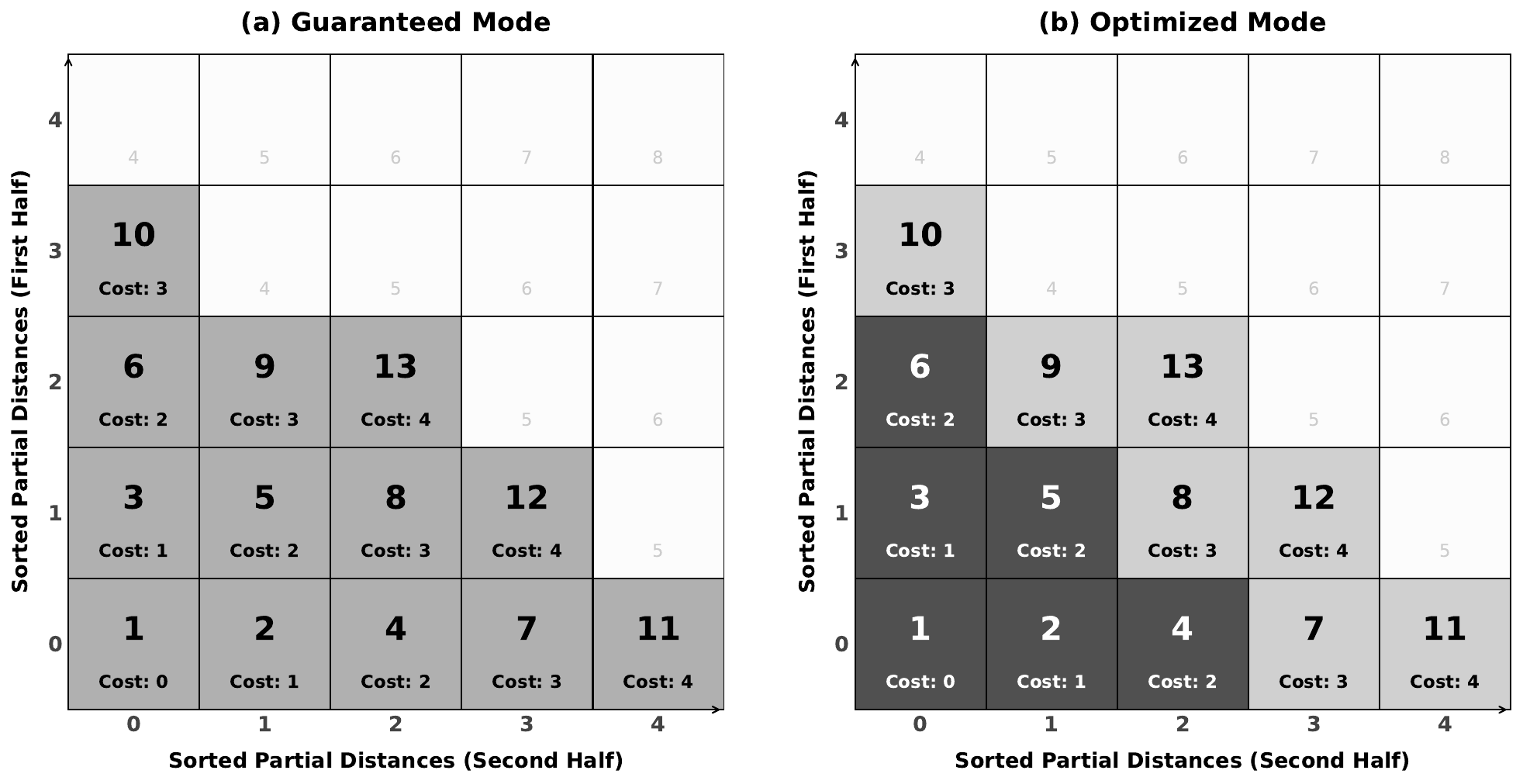}
    \vspace{-0.1in}    \caption{Candidate scoring on a $5 \times 5$ centroid grid. Axes show sorted partial distances $(0, 1, 2, \dots)$ in each subspace half. Cells are visited in ascending order of cost (Ranks 1-13) until sufficient candidates are retrieved. \textbf{(a) Guaranteed Mode:} Uniform weighting ($w=1$) for all candidates (gray). \textbf{(b) Optimized Mode:} The first $k_{size}$ cells (dark gray, Ranks 1-6) with lowest costs receive double weight ($w=2$) to prioritize likely nearest neighbors. The remaining cells (light gray) use standard weight ($w=1$).}
    \label{fig:CRISP-weight}
\end{figure}

\noindent To perform this traversal efficiently, we utilize the \textit{Multi-Sequence Algorithm}~\cite{Babenko2012} adapted for our Inverted Multi-Index structure described in Section ~\ref{sec:csr}. Within each subspace, the vector space is decomposed into two orthogonal halves, each quantized by a separate codebook. Rather than exhaustively evaluating the Cartesian product of all centroid combinations (which scales quadratically with the codebook size), we employ a two-stage strategy. First, we sort the partial squared Euclidean distances for the query against the centroids of each half independently. Second, we utilize a priority queue to incrementally expand candidate cells $(i, j)$ by summing these precomputed partial distances, starting from the closest pair. Figure~\ref{fig:CRISP-weight} depicts the cell traversal order and weighting scheme for a simplified $5 \times 5$ centroid grid. This mechanism ensures that cells are visited strictly in ascending order of their aggregated proximity to the query, allowing the algorithm to terminate the search instantly once the required candidate mass is retrieved, thereby avoiding unnecessary distance computations.

\subsubsection{Execution Pipeline}
CRISP's Dual-Mode execution pipeline
orchestrates three stages of filtering, utilizing AVX-512 vectorization and cache-optimized structures to maintain high performance as dimensionality increases.

\stitle{Stage 1: Candidate Generation.}
The query first accesses the CSR-based inverted index to accumulate collision scores. The linearized nature of the index ensures that this stage maximizes memory bandwidth during the score accumulation. The scoring mechanism adapts to the selected mode:
\begin{itemize}[leftmargin=15pt]
    \item \textbf{Optimized Mode ($\phi=1$):} We apply the Rank-Based Weighted Scoring scheme. Collisions in high-rank cells (top-$k_{\text{size}}$) receive double weight ($w=2$), prioritizing candidates that collide in the most relevant partition regions.
    \item \textbf{Guaranteed Mode ($\phi=0$):} We utilize Binary Collision Counting ($w = 1$ for all collisions). 
\end{itemize}

For each subspace, we first retrieve a small fraction of the dataset (typically between 0.1\% and 6\%) to serve as potential candidates. We then apply a strictness threshold $\tau$, an integer representing the minimum number of subspace collisions a candidate must accumulate to be retained in the final candidate set $\mathcal{C}$.
This multi-stage filtering ensures that the subsequent refinement phase processes only a highly selective fraction of the dataset, minimizing latency while maintaining high recall. A robustness fallback mechanism ensures $|\mathcal{C}| \geq k$ by retrieving the top-ranking items if extreme sparsity leads to candidate underflow.

\stitle{Stage 2: BQ-Accelerated Refinement.}
To optimize the processing order of the candidate set $C$, we utilize a conditional execution path. In Optimized Mode ($\phi=1$), we compute the Hamming distance $\|b(q) - b(x)\|_H$ for all candidates efficiently using AVX-512 intrinsics. Candidates are then sorted by these distances, ensuring that the most promising vectors are evaluated first. This prioritization is vital for the patience mechanism, enabling the search to rapidly identify nearest neighbors and trigger early termination. In Guaranteed Mode ($\phi=0$), however, this refinement step is bypassed entirely, since the theoretical guarantee requires exhaustive verification of all candidates.

\stitle{Stage 3: Mode-Specific Verification.}
The final refinement of candidates is controlled by the user-configurable execution flag $\phi$:
\begin{itemize}[leftmargin=15pt]
    \item \textbf{Optimized Mode ($\phi=1$):} We employ ADSampling combined with Dynamic Termination (Patience). The search estimates distances using subsets of dimensions and aborts if the top-$k$ results remain unchanged for $P$ consecutive verifications (default $P=40\cdot k$). 
    \item \textbf{Guaranteed Mode ($\phi=0$):} The patience mechanism and ADSampling are disabled, as the exact $L_2$ distance needs to be computed for every candidate in set $C$. 
\end{itemize}

\subsection{Complexity Analysis}
We analyze the asymptotic complexity of CRISP across three key dimensions: space efficiency, construction cost, and query latency.

\stitle{Space Efficiency.}
CRISP maintains a strictly linear memory footprint. The raw dataset requires $O(ND)$ storage, while the CSR-organized posting lists add $O(NM)$ for the per-subspace inverted index, yielding a total space complexity of $O(ND + NM)$. This matches the storage complexity of standard subspace partitioning approaches~\cite{DBLP:journals/pacmmod/WeiLLPP25}, placing CRISP in the same compact storage tier as randomized quantization baselines~\cite{DBLP:journals/pacmmod/GaoL24} and ensuring efficient scaling to massive datasets without the prohibitive memory overhead of pointer-heavy structures.

\stitle{Adaptive Construction.}
The construction complexity of CRISP is data-dependent and governed by an adaptive heuristic. 
For isotropic distributions, where dimensions are already weakly correlated, CRISP bypasses rotation to maintain a low $O(ND)$ construction cost, dominated by subspace quantization. The $O(ND^2)$ randomized rotation is triggered selectively only when the spectral check detects high feature correlation, where the alignment gain outweighs the transformation cost. By executing this rotation in-place, we maintain a peak memory footprint of $ND$ floats, effectively halving the memory overhead compared to rigid frameworks \cite{DBLP:journals/pacmmod/GaoL24, DBLP:conf/cvpr/GeHK013} that require $2ND$ space to materialize rotated dataset copies.

\stitle{Query Latency.}
Query processing is primarily dominated by candidate verification, with an asymptotic complexity of $O(|\mathcal{C}| \cdot D)$, where $|\mathcal{C}|$ is the candidate set size. While the worst-case complexity remains the same across modes, our Optimized Mode has significantly lower cost, through ADSampling and patience-based termination, which effectively reduce the number of $D$-dimensional distance calculations.

\section{Theoretical Analysis}
\label{sec:theorem}
Unlike purely heuristic indexing methods, CRISP provides a rigorous lower bound on retrieval quality when operating in Guaranteed Mode.

\begin{theorem} \label{theorem:conditionalrecall}
(Conditional Recall Lower Bound). Let $x^*$ be the true nearest neighbor of query $q$ with single-subspace collision probability $p^*$. Let $\tau$ be the selection threshold defined by $\tau = \alpha \cdot M$. If CRISP is configured in Guaranteed Mode, the probability that $x^*$ is successfully retrieved into the candidate set is lower-bounded by:

\[
P(x^* \in \mathcal{C}) \ge 1 - \exp\left( - \frac{2(Mp^* - \tau)^2}{M} \right)
\]

subject to the condition that the expected collision count $\mu = Mp^*$ strictly exceeds $\tau$.
\end{theorem} 

\begin{proof}
In Guaranteed Mode, the algorithm retrieves every candidate $x$ satisfying $S_{col}(x,q) \ge \tau$. A retrieval failure occurs if and only if the stochastic collision count of the nearest neighbor falls below this threshold:

\[
P_{fail} = P(S_{col}(x^*, q) < \tau)
\]

Assuming the subspaces provide independent evidence, the random variable $S_{col}(x^*, q)$ represents the sum of $M$ independent Bernoulli trials with success probability $p^*$. We apply \textit{Hoeffding's inequality} \cite{hoeffding1963}, which bounds the tail probability of sums of bounded independent random variables. For a sum $S$ with expectation $E[S]$, Hoeffding's inequality states:

\[
P(E[S] - S \ge t) \le \exp\left( - \frac{2t^2}{M} \right)
\]

Setting $S = S_{col}(x^*, q)$ with $E[S] = Mp^*$, and choosing the deviation $t = Mp^* - \tau$ (where $t > 0$ by the theorem's condition):

\[
P(Mp^* - S_{col}(x^*, q) \ge Mp^* - \tau) \le \exp\left( - \frac{2(Mp^* - \tau)^2}{M} \right)
\]

This simplifies to $P(S_{col}(x^*, q) \le \tau)$. Since $P_{fail} = P(S_{col} < \tau) \le P(S_{col} \le \tau)$, the tail bound serves as a valid upper bound on the failure probability. Consequently, the retrieval probability $1 - P_{fail}$ satisfies the stated lower bound.
\end{proof}

\stitle{Validity of Independence Assumption.} While Hoeffding's inequality assumes independence, real-world data often exhibits high inter-dimension correlation which violates this assumption. Our Adaptive Rotation acts as a structural correction: by applying a randomized orthogonal transformation, we decorrelate the feature space and ensure that collisions across subspaces behave as independent events. This ensures that the subspace collisions satisfy the independence property required for our error bound to hold.

\stitle{Interpretation.}The bound demonstrates that the retrieval failure probability decays exponentially with the number of subspaces $M$, offering a stronger guarantee than the polynomial bounds (e.g., Chebyshev) used in prior works \cite{DBLP:journals/pacmmod/WeiLLPP25}. Moreover, it provides theoretical justification for the robust fallback mechanism (Stage 1): when $\tau \ge Mp^*$ (violating the theorem's condition), the theoretical bound becomes vacuous, necessitating the score-based fallback to prevent empty result sets.

\section{Experimental Analysis}\label{sec:exp}

\stitle{Setup.} We experimentally compare CRISP to alternative main-memory indices (RaBitQ \cite{DBLP:journals/pacmmod/GaoL24}, SuCo \cite{DBLP:journals/pacmmod/WeiLLPP25}, OPQ \cite{DBLP:conf/cvpr/GeHK013}, and HNSW \cite{DBLP:journals/pami/MalkovY20}). We systematically tested each method's key parameters to optimize the trade-off between search accuracy, construction time, and peak memory usage.  For OPQ and HNSW, we utilize the implementations provided by the FAISS library~\cite{johnson2019billion}, while RaBitQ and SuCo use their official open-source implementations. All methods are implemented in C++ and compiled using \verb|gcc| (v11.4) on Ubuntu 22.04 LTS. We used the optimization flags \verb|-O3| and \verb|-fopenmp| for all baselines, ensuring multi-threaded execution. All methods were compiled with their optimal available flags: CRISP, SuCo, and FAISS (OPQ/HNSW) with \verb|-march=native| and \verb|-mavx512f|, and RaBitQ with \verb|-mavx2|. The experiments were conducted on a machine equipped with an 11th Gen Intel® Core™ i7-11700K processor (3.60 GHz), 32 GB of RAM, and AVX-512 support.
For each competitor method, we  extended beyond the ranges recommended in the original documentations to thoroughly evaluate performance  across diverse operating regimes.

\textbf{CRISP.} We varied the data subspace structure (tested in five configurations per dataset by adjusting the number of subspaces), the collision ratio $\alpha$ controlling query coverage (spanning $[0.001, 0.06]$), and the minimum-collisions-percentage (in $[0.1, 0.6]$), which controls how many subspace collisions a candidate requires to be selected. We tested both modes of CRISP: Guaranteed Mode (with theoretical guarantees and exhaustive verification) and Optimized Mode (with weighted scoring and early termination).

\textbf{SuCo.} We tested the collision ratio $\alpha$ in the range $[0.02, 0.06]$ and the candidate ratio $\beta$ in $[0.003, 0.006]$, using the same five subspace configurations as CRISP. These ranges were selected to cover and slightly extend the recommended ``best ranges'' of $\alpha \in [0.03, 0.05]$ and $\beta \in [0.001, 0.005]$ suggested by the authors. By extending these boundaries, we ensured that the optimal performance point of SuCo for each dataset was achieved, while being aligned with the paper's guidelines \cite{DBLP:journals/pacmmod/WeiLLPP25}. SuCo was evaluated across all configurations within these parameter intervals, retaining only the Pareto-optimal points along the recall-QPS and recall-construction time trade-off curves.

\textbf{RaBitQ.} We tested multiple cluster sizes ($N_{\text{list}}$) spanning the range from 32 to 4,096, covering the wide variety of dataset cardinalities in our benchmarks and including the $N_{\text{list}} = 4{,}096$ configuration recommended by the authors for million-scale datasets. We further compared fast training (2 k-means iterations) against standard training (20 iterations) to evaluate the robustness of the randomized codebook construction. To ensure a comprehensive Pareto frontier, we tested the number of query probes ($N_{\text{probe}}$) in powers of two, from 1 up to the total cluster count.

\textbf{HNSW.} We tested graph densities up to $M=64$, extending beyond the recommended maximum of 48 to evaluate performance in very-high-dimensional regimes where higher connectivity may be necessary. We also varied construction search depth ($\text{efConstruction} \in \{32, 64, 128\}$) and query search depth ($\text{efSearch} \in \{32, 64, 128, 256\}$) across practical ranges. We used the FAISS implementation \cite{johnson2019billion}.

\textbf{OPQ.} We configured the sub-quantizer count $M$ to match our five subspace partitions, as OPQ requires $M$ to divide the data dimensionality. The encoding precision was tested at 8 bits per sub-vector (the FAISS default) and 6 bits (a supported lower-precision alternative documented for memory-constrained settings). The number of OPQ training iterations was tested at the FAISS default of 50 and a reduced value of 20 to evaluate convergence sensitivity across datasets. A per-run timeout of 1 hour was enforced on index construction; runs exceeding this limit were recorded as failed.

\subsection{Datasets}
To evaluate CRISP across diverse high-dimensional regimes, we selected nine datasets spanning varying cardinalities, modalities, and intrinsic dimensionalities. Table~\ref{tab:datasets} presents the dataset specifications and the sources from which we obtained each dataset and the corresponding query sets.
We characterize each dataset by its Local Intrinsic Dimensionality (LID), estimated via the Maximum Likelihood Estimation (MLE) estimator of Levina \& Bickel \cite{DBLP:conf/nips/LevinaBickel04} on 1,000 randomly sampled queries with $k=100$ neighbors. LID measures the
effective dimensionality of the data's neighborhood structure, independent of
the raw embedding dimension $D$. 

\begin{table}[h]
    \caption{Dataset Characteristics and LID}
    \label{tab:datasets}
    \vspace{-0.05in}
    \centering
    \resizebox{\linewidth}{!}{%
    \begin{tabular}{@{}l l c c c c c@{}}
        \toprule
        \textbf{Dataset} & \textbf{Type} & \textbf{Card. ($N$)} & \textbf{$D$} & \textbf{\#Queries ($Q$)} & \textbf{$k_{\text{size}}$} & \textbf{LID} \\
        \midrule
        \textbf{Gist \cite{gist-dataset}}                 & Image & 1,000,000 & 960   & 1,000  & 100 & 44.81 \\
        \textbf{Simplewiki-OpenAI \cite{datasets}}    & Text  & 260,372   & 3,072 & 1,000  & 100 & 27.49 \\
        \textbf{Trevi \cite{datasets}}                & Image & 90,120    & 4,096 & 1,000  & 100 & 24.85 \\
        \textbf{Ccnews-nomic \cite{datasets}}         & Text  & 495,328   & 768   & 1,000  & 100 & 22.94 \\
        \textbf{Agnews-mxbAI \cite{datasets}}         & Text  & 769,382   & 1,024 & 1,000  & 100 & 20.92 \\
        \textbf{Imagenet \cite{datasets}}             & Image & 1,281,167 & 640   & 1,000  & 100 & 20.54 \\
        \textbf{Gooaq-distilroberta \cite{datasets}}  & Text  & 1,475,024 & 768   & 1,000  & 100 & 17.09 \\
        \textbf{Fashion-MNIST \cite{datasets}}        & Image & 60,000    & 784   & 10,000 & 100 & 15.26 \\
        \textbf{MNIST \cite{datasets}}                & Image & 69,000    & 784   & 200    & 100 & 14.07 \\
        \bottomrule
    \end{tabular}
    }
\end{table}

As shown in Table~\ref{tab:datasets}, \textbf{Gist} is the most challenging dataset ($\text{LID} \approx 44.8$), far exceeding its moderate-to-high $D=960$. \textbf{Simplewiki-OpenAI} and \textbf{Trevi} follow ($\text{LID} \approx 27.5$ and $24.9$), with Trevi's high $D=4{,}096$ largely explained by inter-dimensional correlations rather than true complexity. \textbf{Ccnews-nomic}, \textbf{Agnews-mxbAI}, and \textbf{Imagenet} occupy the mid-range ($\text{LID} \approx 20$--$23$), reflecting typical dense retrieval embeddings,
while \textbf{Gooaq-distilroberta} sits slightly lower ($\text{LID} \approx 17.1$),
despite its large cardinality of 1.47M vectors. \textbf{MNIST} and \textbf{Fashion-MNIST} anchor the low end ($\text{LID} \approx 14$--$15$), consistent with their well-known manifold structure.
\begin{figure*}[t]
    \centering
    \includegraphics[width=0.95\textwidth]{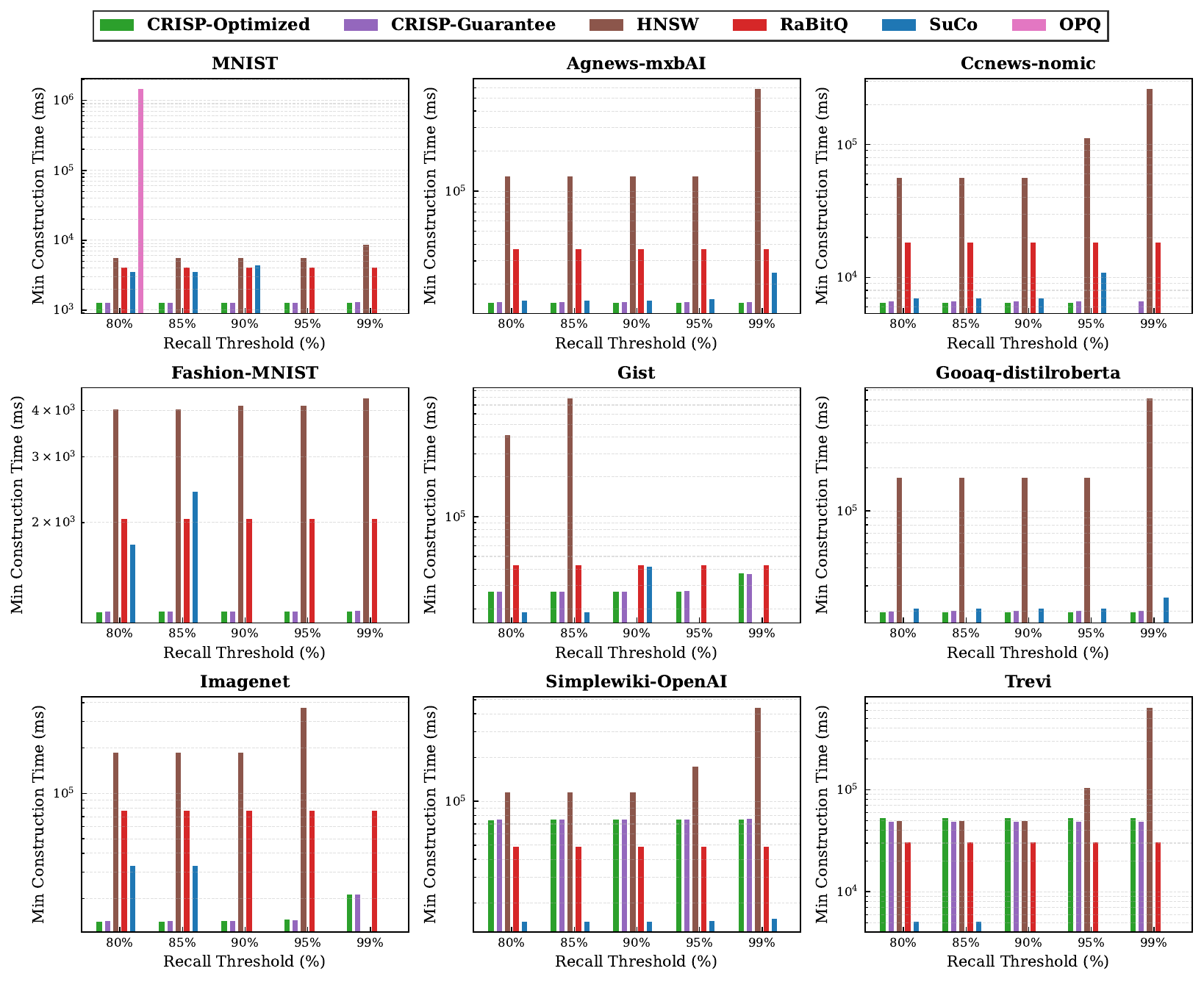}
    \vspace{-0.1in}
    \caption{Minimum index construction time required to reach specific Recall@100 thresholds (80\%, 85\%, 90\%, 95\%, 99\%) across nine benchmark datasets (log-scale y-axis). Missing bars indicate that a method could not achieve the required recall threshold with the tested configurations. Lower is better.}

    \label{fig:recall_build}
\end{figure*}

\subsection{Indexing Scalability vs. Retrieval Quality}
\label{sec:scalability}

We evaluate the cost of index construction by computing the Recall@100 vs.\ construction time Pareto frontier for each method. Figure~\ref{fig:recall_build} shows 
the lowest build time required for a parameters configuration that achieves (at least) the recall level at the x-axis.
This experiment evaluates the trade-off between retrieval accuracy and indexing cost.

\stitle{Construction Efficiency.}
CRISP's construction cost remains nearly constant across recall levels. On \textbf{Trevi} ($D\!=\!4096$), CRISP reaches recall from 85\% to 99.5\% requiring a 49s-53s construction cost, as its index build is a fixed-cost operation independent of the search-time parameters (candidate ratio, minimum score). SuCo achieves the fastest build time at $\sim$5s owing to its simple partitioning scheme, though it reaches a maximum recall of only 86\% on this dataset. By contrast, HNSW's construction cost grows sharply with recall: 49s at 92\% recall, rising to 634s at 99.4\%-a $13\times$ increase. RaBitQ's construction time is fixed at ${\sim}30$s across all recall levels, as its build cost depends on the number of clusters while the recall-controlling parameter (NProbe) is search-time only and does not
affect construction.
On \textbf{Imagenet} ($D\!=\!640$), CRISP builds its index in 14s-15s across all recall levels up to 98.6\%. HNSW's Pareto frontier appears only above 93\% recall on this dataset, requiring 187-395s to build depending on graph connectivity. RaBitQ reaches 99.3\% recall but requires 77s, while SuCo maintains a low build cost (33s) at the expense of an 86\% recall ceiling. In \textbf{Gooaq-distilroberta}, RaBitQ could not be evaluated due to its prohibitive memory requirements during construction.

On \textbf{MNIST} ($D\!=\!784$), where all methods achieve high recall, CRISP's construction completes in $\sim$1.3s for up to 99.9\% recall. SuCo is comparably fast at 3.5-4.4s, though its recall plateaus at $\sim$92\%. HNSW requires 5.5-8.6s, and RaBitQ $\sim$4s. Even on this well-structured dataset, CRISP's construction is $4{-}7\times$ faster than HNSW and $\sim$$3\times$ faster than RaBitQ at comparable recall levels. OPQ only appears in Figure~\ref{fig:recall_build}  on MNIST, as it requires approximately 24 minutes, over $1{,}100\times$ longer than CRISP at comparable recall. On all remaining datasets, OPQ either timed out or failed to reach the 80\% recall threshold under any tested configuration. Across all datasets, CRISP's flat construction profile stems from its design: the build cost is dominated by a single pass over the data for subspace encoding, with a spectral check that adds negligible overhead when rotation is unnecessary (close to 2.5\% of the total construction cost). SuCo benefits from a similarly lightweight build but is limited by its inability to reach high recall on datasets with high spectral energy concentration (high CEV), such as \textbf{Gist} and \textbf{Fashion-MNIST}, where independent subspace partitioning fails to capture the true underlying data distribution. HNSW's graph construction scales superlinearly with connectivity (M) and candidate pool size (efConstruction), while RaBitQ's k-means training cost grows with the number of clusters.

\stitle{Memory Efficiency.} Table~\ref{tab:memory_footprint} reports the search-phase Resident Set Size (RSS) for each method at the highest-recall Pareto configuration of each method. While CRISP and SuCo both share an asymptotic footprint of $O(ND + NM)$, CRISP consistently requires $\approx$1.85$\times$ less RAM in practice. This efficiency is primarily attributed to our use of a cache-coherent Compressed Sparse Row (CSR) inverted index, which eliminates the metadata overhead (i.e. fragmentation and pointers) of conventional hash-based inverted lists.
HNSW exhibits a memory complexity of $O(ND + NM_{HNSW} \log N)$, where $M_{HNSW}$ denotes the number of bidirectional links per node; consequently, increasing graph connectivity to achieve high recall on challenging datasets linearly inflates both graph storage and construction-time memory overhead. RaBitQ stores the full-precision dataset alongside binary quantized codes of dimension $B$ and $K_{q}$ cluster centroids, yielding a footprint of $O(ND + NB + K_{q}B)$, while also incurring significant peak memory usage during the cluster-training phase.
As already mentioned, RaBitQ's construction run out of memory on  \textbf{Gooaq-distilroberta}.

In summary, CRISP achieves the lowest 
memory footprint across all datasets tested -- comparable to HNSW and RaBitQ in asymptotic complexity, and well below SuCo despite a similar underlying algorithmic structure -- due to careful memory management.

\begin{table}[t]
\centering
\caption{Peak search-phase RSS (GB) per method and dataset.}
\label{tab:memory_footprint}
\vspace{-0.05in}
\small 
\setlength{\tabcolsep}{3pt}
\resizebox{\columnwidth}{!}{%
\begin{tabular}{lcccc}
\toprule
\textbf{Dataset} & \textbf{CRISP} & \textbf{SuCo} \cite{DBLP:journals/pacmmod/WeiLLPP25} & \textbf{RaBitQ} \cite{DBLP:journals/pacmmod/GaoL24}  & \textbf{HNSW} \cite{DBLP:journals/pami/MalkovY20} \\
Agnews-mxbAI      & \textbf{3.57} & 6.6  & 6.41 & 6.47 \\
Ccnews-nomic     & \textbf{1.73} & 3.2  & 3.11 & 3.24 \\
Gist        & \textbf{4.35} & 8.03 & 7.80 & 7.78 \\
Gooaq-distilroberta       & \textbf{5.13} & 9.18 & --   & 9.44 \\
Imagenet    & \textbf{3.79} & 6.78  & 6.67 & 7.08 \\
Simplewiki-OpenAI  & \textbf{3.34} & 6.4  & 6.59 & 6.38 \\
Trevi       & \textbf{1.50} & 2.95  & 3.19 & 2.96 \\
MNIST       & \textbf{0.24} & 0.46  & 0.47 & 0.47 \\
Fashion-MNIST      & \textbf{0.25} & 0.46  & 0.48 & 0.45 \\
\midrule
Space complexity & $O(ND{+}NM)$ & $O(ND{+}NM)$ & $O(ND{+}NB{+}K_{q}B)$ & $O(ND{+}NM_{HNSW}\log N)$ \\
\bottomrule
\end{tabular}
}
\end{table}

\begin{figure*}[t]
    \centering
    \includegraphics[width=0.9\textwidth]{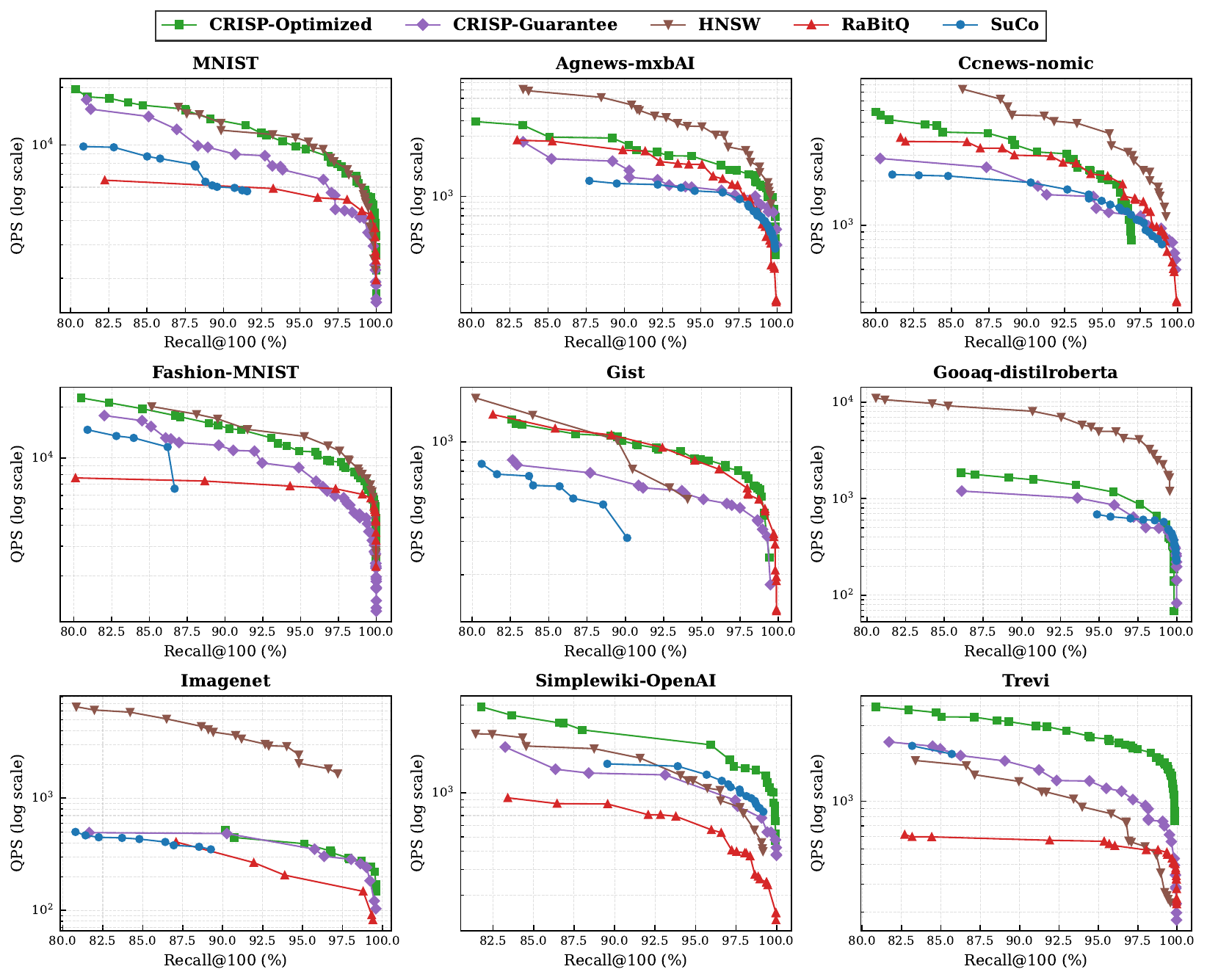}
    \vspace{-0.1in}
    \caption{Recall@100 vs.\ QPS Pareto frontiers across nine benchmark datasets (log-scale y-axis). Each subplot shows the optimal throughput-accuracy trade-off for each method. Higher and further right is better.}
    \label{fig:recall_qps}
\end{figure*}

\subsection{Retrieval Performance}
\label{sec:retrieval}

The primary trade-off in ANN search index design is between query throughput, measured in Queries Per Second (QPS), and retrieval accuracy (Recall@100). 
Figure~\ref{fig:recall_qps} demonstrates that CRISP achieves competitive or superior performance across diverse dataset categories. We evaluate both modes of CRISP denoted as CRISP-Optimized and CRISP-Guarantee. 
OPQ is omitted as it either timed out or failed to reach the 85\% recall threshold on any dataset.

CRISP's advantage is most pronounced in the highest-dimensional settings. On \textbf{Trevi} ($D\!=\!4096$), CRISP-Optimized outperforms all methods by far, being $2.95\times$ faster than the runner-up at 95\% recall (2,463 vs.\ HNSW's 834 QPS) and $6.6\times$ at 99\% recall (1,751 vs.\ HNSW's 267 QPS). Remarkably, even CRISP-Guarantee outperforms HNSW on Trevi, delivering 1,206 QPS at 95\% recall and 746 QPS at 99\% recall. On \textbf{Simplewiki-OpenAI} ($D\!=\!3072$), CRISP-Optimized achieves 2,137 QPS at 95\% recall compared to HNSW's 1,080 QPS and RaBitQ's 559 QPS. At 99\% recall, CRISP-Optimized maintains 1,319 QPS against HNSW's 458 QPS and RaBitQ's 245 QPS, while CRISP-Guarantee still achieves a highly competitive 673 QPS. In both cases, CRISP's contiguous CSR memory layout and sequential scan patterns avoid the pointer-chasing overhead that penalizes graph traversal when vector data exceeds the L2 cache.

On \textbf{Imagenet} ($D\!=\!640$), performance varies with the recall threshold. At lower recall thresholds ($\leq$95\%), HNSW leads in throughput. However, as the recall target increases, CRISP-Optimized delivers 243 QPS at 99\% recall, i.e., $2.7\times$ higher than RaBitQ's 91 QPS, while HNSW does not reach this threshold. Notably, CRISP-Guarantee matches this performance, delivering 241 QPS at 99\% recall, demonstrating that strict bounds do not sacrifice throughput at high recall targets. CRISP is also the only method to achieve $\geq$99.5\% recall on this dataset.

On \textbf{Agnews-mxbAI} ($D\!=\!768$), CRISP is the runner up after HNSW across most thresholds; at 99.5\% recall, CRISP-Optimized delivers 989 QPS against HNSW's 1,021 QPS, while CRISP-Guarantee provides 745 QPS. Both CRISP's modes consistently outperform RaBitQ by up to $2\times$.
On \textbf{Ccnews-nomic} ($D\!=\!768$) and \textbf{Gooaq-distilroberta} ($D\!=\!768$), HNSW maintains a throughput advantage at all recall levels, with CRISP being the runner-up at recall levels less than 95\% and 99\%, respectively.

On well-structured datasets, \textbf{MNIST} ($D\!=\!784$) and \textbf{FashionMNIST},
CRISP leads or comes closely second after HNSW at all recall levels. 
CRISP-Guarantee remains highly performant, outperforming RaBitQ and SuCo at most recall levels.

On \textbf{Gist} ($D\!=\!960$), a dataset with strong inter-dimensional correlation, most methods struggle to reach high recall. Neither HNSW nor SuCo achieve 95\% recall at any practical throughput level (Figure~\ref{fig:recall_qps}). CRISP-Optimized reaches 95\% recall at 799 QPS compared to RaBitQ's 718 QPS, and at 97\% recall delivers 707 QPS versus RaBitQ's 571 QPS. CRISP-Guarantee remains competitive achieving 499 QPS at 95\% recall and 450 QPS at 97\% recall.

\subsection{Ablation Study}
\label{sec:ablation}

To evaluate the impact of the design choices and parameter selection within CRISP, we conduct an ablation study on 
four datasets that span a diverse range of intrinsic dimensionalities, spectral structures, and semantic domains. Specifically, \textbf{Fashion-MNIST} and \textbf{Gist} both exhibit high spectral energy concentration (CEV $= 0.94$ and $0.91$ respectively), though they differ substantially in intrinsic dimensionality ($15.3$ vs. $44.8$), making them complementary stress tests for the adaptive rotation mechanism; \textbf{Simplewiki-OpenAI} covers an extremely high-dimensional embedding space ($D = 3{,}072$, $\text{LID} \approx 27.5$); and \textbf{Ccnews-nomic} represents a near-isotropic distribution where variance is spread uniformly across dimensions (CEV $= 0.77$, $\text{LID} \approx 22.9$).

\stitle{Sensitivity Analysis of the Adaptive Rotation Threshold.} The effectiveness of our subspace partitioning depends on how well the dataset aligns with the principal coordinate axes. To determine the optimal threshold for triggering global rotation, we evaluated a range of $\tau_{cev}$ values.
As illustrated in Figure \ref{fig:ablation}, 
$\tau_{cev} = \mathbf{0.85}$ yields the strongest overall performance, 
providing an effective decision boundary for discriminating between datasets that benefit from orthogonal rotation and those that do not. 
For datasets whose intrinsic structure forms tight local manifolds within the high-dimensional space, this threshold correctly triggers the transformation. Suppressing rotation on these datasets by raising the threshold 
yields a substantial degradation in retrieval quality.
Conversely, for datasets where variance is spread across many dimensions, orthogonal rotation provides negligible accuracy gains while imposing a non-trivial query-time overhead. 

\begin{figure}[h]
    \centering
    \includegraphics[width=0.9\linewidth]{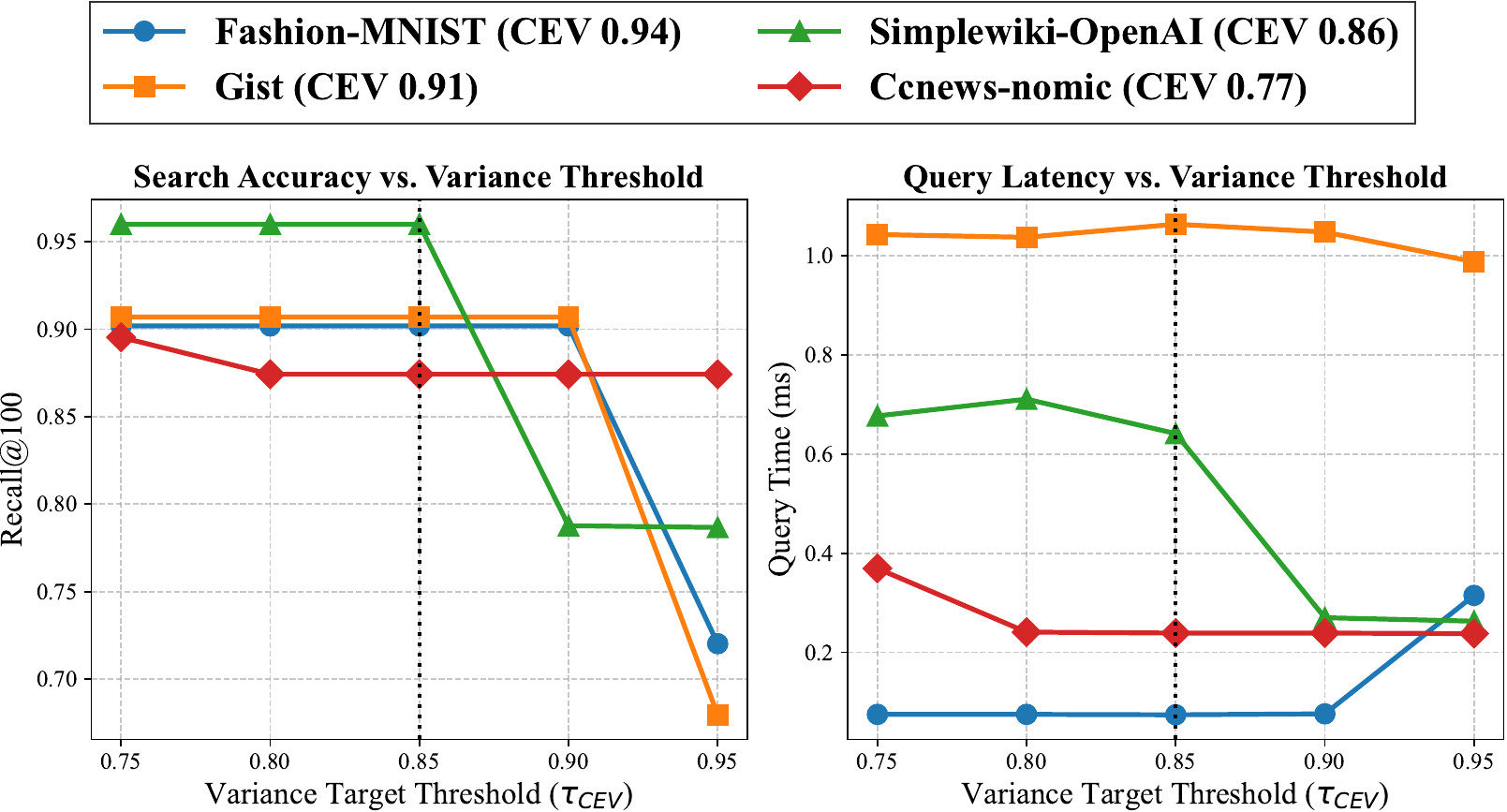}
    \caption{Impact of the Variance Target Threshold ($\tau_{cev}$) on (a) Search Accuracy (Recall@100) and (b) Query Latency (ms) across the representative datasets.}
    \label{fig:ablation}
\end{figure}

\stitle{Sensitivity Analysis of the Patience Factor.} The early termination mechanism of CRISP-Optimized
trades search accuracy for query throughput. We evaluated the patience factor $P$ used in this tradeoff, i.e., 
the number of consecutive candidates examined without an update to the top-$k$ results
and testing for early termination. As shown in Figure \ref{fig:patience_ablation}, our results identify 40 as the most effective value for this parameter across the tested datasets.
Observe that increasing 
$P$ from 20 to 40 improves recall significantly in almost all datasets; however, further increasing the factor to 60 brings no additional benefit. On the other hand, as expected, query latency  linearly increases with $P$, yielding $P=40$ the best empirical value.

\begin{figure*}[h]
    \centering
    \includegraphics[width=0.9\linewidth]{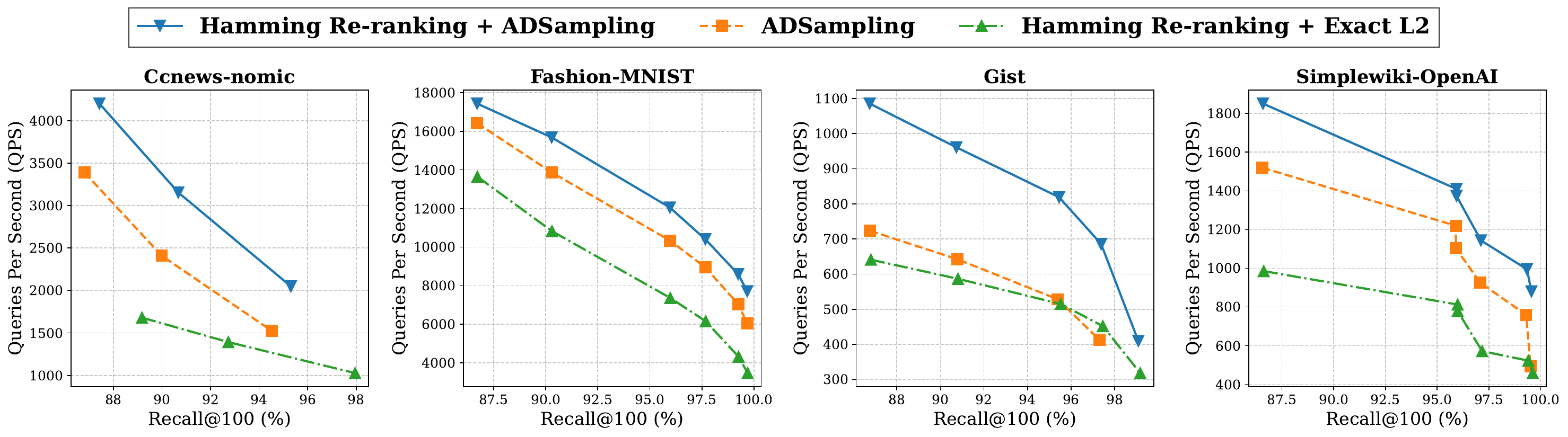}
    \caption{Ablation study of filtering subcomponents showing the Recall-QPS Pareto frontier for (a) Ccnews-nomic, (b) Fashion-MNIST, (c) Gist, and (d) Simplewiki-OpenAI.}
    \vspace{-0.01in}
    \label{fig:subcomponents_ablation}
\end{figure*}

\stitle{Impact of the Multi-Stage Filtering Pipeline.} The CRISP query pipeline employs two components to accelerate the final verification phase: binary Hamming re-ranking for candidate prioritization and ADSampling for approximate distance computation. To evaluate them, we compare the full pipeline against two configurations: one using Hamming re-ranking with exact L2 distances (ADSampling disabled), and one using ADSampling without prior Hamming re-ranking. As shown in Figure~\ref{fig:subcomponents_ablation}, both components are essential for achieving the optimal Recall--QPS Pareto frontier.

The results show that ADSampling is the primary driver of query throughput; e.g., for \textbf{Gist} at $90\%$ recall, enabling ADSampling increases throughput by over $2\times$ relative to exact L2 distance computation. 
These gains come at negligible cost to recall, as ADSampling's adaptive pruning effectively skips redundant dimensions during distance computation.
At the same time,
the effectiveness of patience-based early termination, depends critically on the evaluation order of the candidates. Switching off Hamming re-ranking
leads to measurable performance degradation. On \textbf{Gist}, removing the Hamming re-ranking stage while retaining ADSampling results in a $22\%$ drop in throughput
at equivalent recall levels. This confirms the effectiveness of binary Hamming codes  as an effective low-cost proxy for Euclidean proximity.

\begin{figure}[h]
    \centering
    \vspace{-0.05in}
\includegraphics[width=0.9\linewidth]{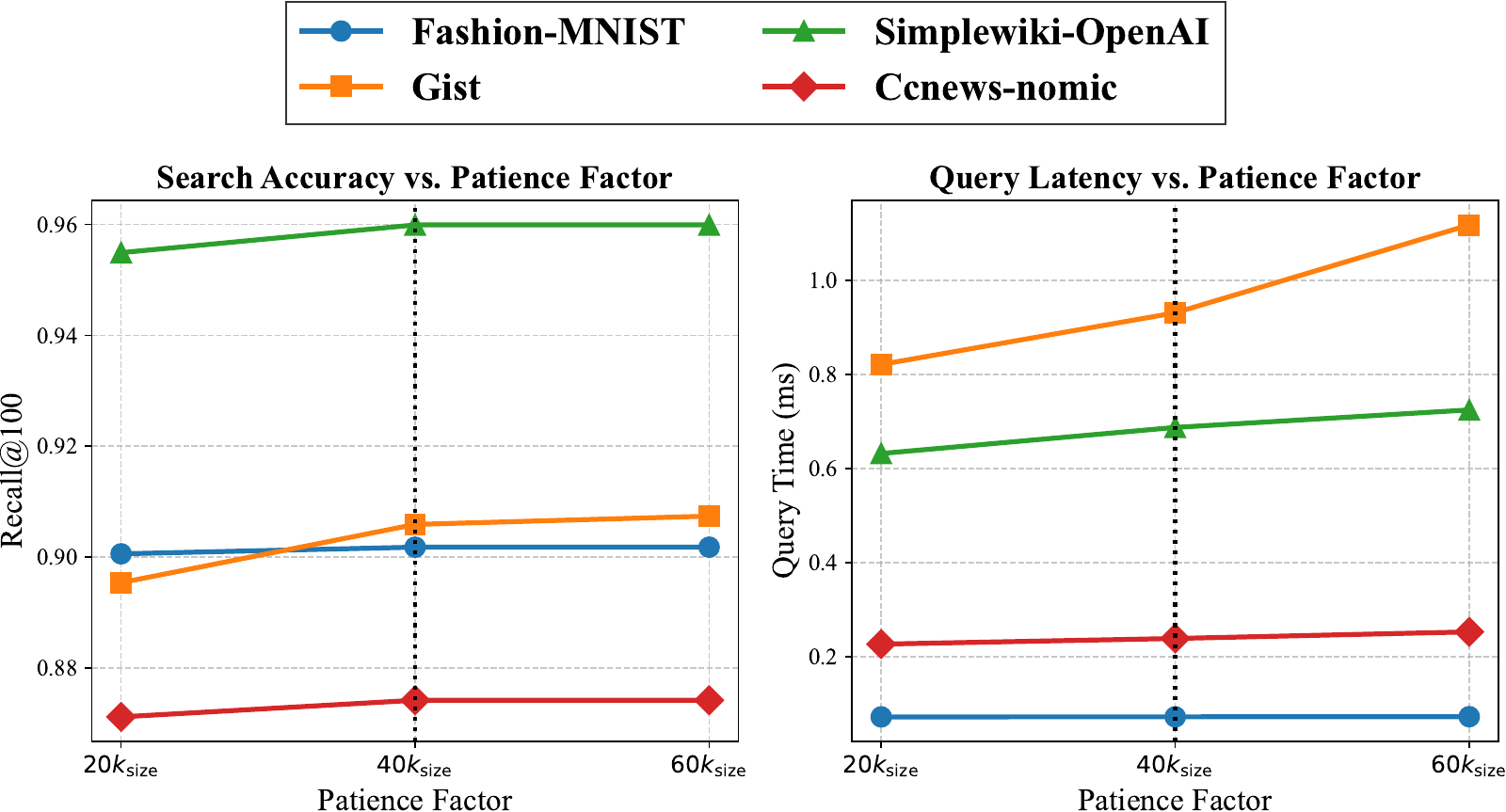}
    \caption{Impact of the Patience Factor on (a) Search Accuracy (Recall@100) and (b) Query Latency (ms) across the representative datasets.}
    \vspace{-0.1in}
\label{fig:patience_ablation}
\end{figure}

\subsection{Summary of Experimental Analysis}

Across all benchmarks, CRISP demonstrates its most decisive advantages in extreme-dimensional regimes. On datasets with $D \geq 3072$ (\textbf{Trevi}, \textbf{Simplewiki-OpenAI}), CRISP-Optimized delivers up to $6.6\times$ higher throughput than HNSW at 99\% recall, while maintaining construction times that are an order of magnitude lower. At $D \approx 960$ (\textbf{Gist}), where strong inter-dimensional correlation causes both HNSW and SuCo to fail at reaching high recall entirely, CRISP remains the only method capable of exceeding 97\% recall at practical throughput levels. Beyond retrieval performance, CRISP consistently maintains the lowest search-phase memory footprint across all tested datasets, requiring $\approx$1.85$\times$ less RAM than SuCo despite a similar algorithmic structure, and remaining more compact than both RaBitQ and HNSW. In terms of construction, CRISP exhibits superior or competitive performance compared to all methods, except on {\bf Simplewiki-OpenAI} and {\bf Trevi}, where RaBitQ and SuCo are faster to construct, but achieve low QPS rates compared to CRISP.

\looseness=-1
CRISP's limitations emerge in lower-dimensional settings ($D \leq 768$), where graph-based methods remain highly competitive. Although CRISP consistently outperforms RaBitQ and SuCo across these datasets, and closely matches HNSW on \textbf{MNIST} and \textbf{Fashion-MNIST}, it falls behind on \textbf{Ccnews-nomic} and \textbf{Gooaq-distilroberta}, where HNSW sustains a clear throughput advantage at most
recall levels. On \textbf{Ccnews-nomic} specifically, CRISP-Optimized reaches up to 97\% recall. Similarly, on \textbf{Imagenet} ($D = 640$), HNSW leads at moderate recall thresholds; however, above 97.5\% recall, CRISP-Optimized becomes the best performing method, and is the only one capable of reaching 99.5\% recall on this dataset.

\section{Conclusions and Future Work}

In this work, we presented CRISP, a high-performance indexing framework for very high dimensional data that combines correlation-aware preprocessing with cache-efficient data structures and vectorized query processing. While our results demonstrate strong performance across diverse benchmarks, the following directions remain open for extending the framework:

\stitle{Adaptive Subspace Decomposition.} CRISP currently splits the feature space into subspaces of equal size (e.g., 128 dimensions into 8 subspaces of 16 dimensions each). This uniform split ignores the fact that some dimensions carry more information than others. An unexplored idea is to adjust subspace sizes based on variance: allocate fewer dimensions to high-variance regions where finer quantization is needed, and group more dimensions together in low-variance regions. This would reduce quantization error on skewed datasets without increasing memory usage.

\stitle{Partial Variance Redistribution.} Currently, CRISP applies a global rotation to the entire feature space when correlation is detected. A promising optimization is Block-Wise Variance Redistribution, where rotation is applied only to a specific subset of dimensions exhibiting high spectral energy concentration. This approach redistributes variance locally among a target set of subspaces rather than across the entire vector, reducing the preprocessing complexity from $O(N D^2)$ to $O(N D \cdot d_{\text{sub}})$, where $d_{\text{sub}} \ll D$. Hence, finer-grained adaptability can be achieved in mixed-distribution datasets in which only particular feature subspaces are correlated.



\bibliographystyle{ACM-Reference-Format}
\bibliography{sample,liburl}

\end{document}